\documentclass{article}

\usepackage{arxiv}

\usepackage[utf8]{inputenc} 
\usepackage[T1]{fontenc}    
\usepackage{hyperref}       
\usepackage{url}            
\usepackage{booktabs}       
\usepackage{amsfonts}       
\usepackage{nicefrac}       
\usepackage{microtype}      
\usepackage{lipsum}

\usepackage{cite}
\usepackage{amsmath,amssymb,amsfonts}
\newenvironment{pf}{\begin{proof}}{\end{proof}}

\newenvironment{thm}{\begin{theorem}}{\end{theorem}}
\usepackage{graphicx}
\usepackage{algorithm}
\usepackage[noend]{algpseudocode}
\usepackage{hyperref}
\usepackage{textcomp}
\usepackage{threeparttablex} 
\usepackage{adjustbox}
\usepackage{mathtools}
\usepackage{caption}
\usepackage{float}
\usepackage{stmaryrd}
\usepackage{epstopdf}
\usepackage{multirow}
\usepackage{pifont}
\usepackage{caption}
\usepackage{subcaption}
\usepackage{xcolor}

\newcommand{\RD}[1]{{#1}}

\newcommand{\R}{{\mathbb{R}}}

\newcommand{\G}{{\mathcal{G}}}
\newcommand{\F}{{\mathcal{F}}}

\newcommand{\ds}{{\ \mathsf{d}s}}
\newcommand{\Id}{{\mathbb{I}}}
\newcommand{\Rl}{{\mathfrak{R}}}
\newcommand{\qed}{\hfill $\square$}

\newtheorem{theorem}{Theorem}[section]

\newtheorem{assumption}{Assumption}

\newtheorem{definition}[theorem]{Definition}
\newtheorem{lemma}[theorem]{Lemma}
\newtheorem{remark}[theorem]{Remark}
\newtheorem{problem}[theorem]{Problem}
\newtheorem{proof}[theorem]{Proof}

\title{Temporal Logic Resilience for Continuous-time Systems}

\author{
Ratnangshu Das$^{1}$, 
Negar Monir$^{2}$, 
Youssef Ait Si$^{3}$, 
Adnane Saoud$^{3}$, 
Sadegh Soudjani$^{4}$, 
Pushpak Jagtap$^{1}$ 
\\[1ex]
$^{1}$Indian Institute of Science, Bangalore, India \\
\texttt{\{ratnangshud, pushpak\}@iisc.ac.in} \\
$^{2}$Newcastle University, United Kingdom \\
\texttt{S.Seyedmonir2@newcastle.ac.uk} \\
$^{3}$University Mohammed VI Polytechnic, Benguerir, Morocco \\
\texttt{\{youssef.aitsi, adnane.saoud\}@um6p.ma} \\
$^{4}$Max Planck Institute for Software Systems, Germany, and University of Birmingham, United Kingdom \\
\texttt{sadegh@mpi-sws.org}
}

\begin{document}
\maketitle
\setcounter{footnote}{0}
\begin{abstract}                          
In this paper, we present a novel framework for quantifying a lower bound on resilience in continuous-time (non)linear systems subject to external disturbances while ensuring satisfaction of signal temporal logic specifications. Unlike robustness, which evaluates how well a system satisfies a specification under a given disturbance, resilience measures the maximum disturbance a system can tolerate from a given initial state while maintaining specification satisfaction. We first derive bounds on the perturbed trajectories and then use them to formulate a computational method based on scenario optimization to efficiently compute the maximum admissible disturbance. We validate our approach through case studies, including dc motor, temperature regulation, a nonlinear numerical example, and a vehicle collision avoidance case.
\end{abstract}

\section{Introduction}
Cyber-physical systems are widely used in safety-critical applications such as autonomous vehicles, robotic systems, and industrial automation, where they must operate reliably despite disturbances, such as modeling errors and external perturbations. These uncertainties can cause deviations from expected behavior; therefore, it is important to ensure resilience, i.e., the system should continue to satisfy its desired specifications even under disturbances. A well-established formal framework for specifying such complex time-sensitive tasks is through Signal Temporal Logic (STL) \cite{STL_OdedMaler}. However, a fundamental question remains: How much disturbance can a system tolerate while still satisfying an STL specification?
\RD{Unlike fault-tolerance, which focuses on detecting and recovering from specific failures, we adopt a disturbance-tolerance view of resilience, where the goal is to quantify the maximum disturbance magnitude under which the system is guaranteed to satisfy its STL specification.}

As explored by \cite{STL_robust} and \cite{ MITL_robust}, research in temporal logic-based verification has focused on robustness, which quantifies how well a system satisfies a given specification under a fixed disturbance. Various approaches have been proposed to evaluate robustness, including discounting modalities by \cite{almagor2014discounting}, which reduce the influence of distant events, and averaging modalities by \cite{bouyer2014averaging}, which extend classical temporal logic operators with min-, max- and long-run average computations. 
\RD{
Several computational techniques have been developed, with applications in runtime verification \cite{runtimeverification} and temporal logic falsification \cite{temporallogicfalsification}. Recent monitoring-based works, such as \cite{zhong2021extending, roohi2018parameter}, extend STL semantics to monitor uncertain signals, by assuming that deviation sets are already given and by focusing on verifying whether the specification holds under such uncertainties.}
\RD{
Recent works such as \cite{chen2022stl} apply STL-based analysis to characterize resilience properties, including recoverability and durability. While previous studies focus on evaluating satisfaction under a fixed disturbance, our work instead quantifies the admissible disturbance set itself, following the resilience notion introduced by \cite{resilience_water} and \cite{Resilience01}.}

\RD{
Control-oriented approaches have also studied disturbance bounds for STL satisfaction, for example using control barrier functions to characterize admissible disturbances for closed-loop systems \cite{akella2021disturbance}. Although general, such approaches may introduce additional conservatism and are typically restricted to STL fragments that can be encoded through barrier-function constraints.}

Beyond temporal logic verification, resilience has been widely studied in critical infrastructure \cite{resilience_critical}, IT systems \cite{resilience_IT}, and water resource recovery facilities \cite{resilience_water}, where it is commonly interpreted as the ability to maintain acceptable performance under disturbances. However, resilience in cyber-physical systems under STL specifications remains largely unexplored.

\cite{Resilience00} introduced resilience for discrete-time systems using linear temporal logic over finite traces \cite{LTLF}. Extending this to continuous-time systems is a significant challenge, as disturbances evolve continuously and affect system behavior in more complex ways. Unlike discrete-time settings, where disturbances can be modeled as discrete transitions, continuous-time disturbances require analysis over an infinite space of possible deviations. Additionally, STL includes explicit timing constraints, which require analyzing infinitely many possible trajectories arising from continuous disturbance signals. These factors make STL-based resilience analysis considerably more complex and computationally demanding.

In this paper, we develop a framework to compute a lower bound on resilience for continuous-time (non)linear systems under STL specifications. \RD{To understand how disturbances affect system trajectories, we first derive tight bounds on perturbed trajectories, improving over element-wise absolute matrix bounds \cite{Absolute_kofman} and Grönwall-type estimates \cite{akella2021disturbance}.} Using these bounds, we compute the maximum admissible disturbance via scenario optimization \cite{esfahani2014performance} and ensure theoretical correctness. The framework is validated through several case studies, demonstrating its applicability across linear and nonlinear systems.

\section{Preliminaries} \label{sec:prob}
\textbf{Notations:} 
The symbol $|\cdot|$ denotes the element-wise magnitude of a matrix or a vector,  $\|\cdot\|$ is its 2-norm, and $\|\cdot\|_{\infty}$ is its sup norm.
For $A, B \in \R^{m \times n}$, the inequalities, $A \preceq B$ (and $A \succeq B$) represents $A_{ij} \leq B_{ij}$ (and $A_{ij} \geq B_{ij}$), for all $i \in [1;m]$ and $j \in [1;n]$.
The set $\mathcal{B}_{\infty}(\varepsilon) := \{v \in \R^{n} | -\varepsilon \preceq v \preceq \varepsilon\}$ is a hyperrectangle in $\R^n$ centered at origin. 
The symbol $\oplus$ denotes the Minkowski sum, which for two sets $X, Y \subset \R^n$ is defined as $X \oplus Y := \{ x + y \mid x \in X, y \in Y \}$.
$\Rl(\lambda)$ is the real part of $\lambda$.
All other notations follow standard mathematical conventions.

\subsection{System Definition}
A continuous-time system is defined by the tuple $\Sigma_{nl} = (X, D, f)$, where $X \subseteq \R^n$ is the state space, $D \subset \R^n$ is the disturbance space, and $f: \R^{n} \rightarrow \R^{n}$ define the system dynamics. The evolution of the state $x$, under additive disturbance $d$ is given by the differential equation:
\begin{align}
    \Sigma_{nl}: \dot{x} = f(x) + d, \ x: \R_0^+ \rightarrow X, \ d: \R_0^+ \rightarrow D.  \label{eqn:sysdyn_nl}
\end{align}
We assume existence of a unique solution under a continuous signal $d$.
\RD{We consider the set of admissible disturbances to be given by an infinity-norm ball centered at the origin,
$D := \mathcal{B}_\infty(\varepsilon),$
where $\varepsilon > 0$ denotes the disturbance radius. That is, $d:\R_0^+ \rightarrow \mathcal{B}_\infty(\varepsilon)$ represents a bounded disturbance signal satisfying $\|d(t)\|_\infty \le \varepsilon$ for all $t \ge 0$.}

Throughout most of this paper, we focus on continuous-time linear systems, where the evolution of the system $\Sigma_{l} = (X, D, A)$ is governed by the system matrix $A \in \R^{n \times n}$, leading to linear dynamics
\begin{align}
    \Sigma_{l}: \dot{x} = Ax + d, \ x: \R_0^+ \rightarrow X, \ d: \R_0^+ \rightarrow D. \label{eqn:sysdyn_lin}
\end{align}
In Section \ref{sec:nonlin}, we extend the analysis for linear systems to nonlinear systems, deriving resilience metrics under some additional assumptions.

{For a given initial state $x_0 \in X$, the state trajectory evolving over the time interval $[0,t_f]$, with $t_f \in \R_0^+$, under the influence of disturbance $d$ is denoted as $\xi_d(x_0): [0,t_f]\rightarrow \R^n$. The trajectory’s value at time $t$ is represented by $\xi_d(x_0,t)$, where $\xi_d(x_0,t) \in X$ for all $t \in [0,t_f]$.}

\subsection{Signal Temporal Logic}
Signal Temporal Logic (STL) \cite{STL_OdedMaler} is a formal language used to specify the spatial, temporal, and logical properties of continuous-time signals. The set of STL formulae can be recursively expressed using predicates $\mathsf{p}$. Consider the predicate function $h:\R^n \rightarrow \R$, then the predicate $\mathsf{p}$ is true if $h(s) \geq 0$, (i.e.,
$\mathsf{p} := \top$, if $h(s) \geq 0$), and it is false if $h(s) < 0$.
For example, a predicate $\mathsf{p}:=(s \geq 1)$ can be represented by the predicate function $h(s):=s-1$. 

An STL formula $\psi$ is recursively defined using predicates, Boolean logic, and temporal operators:
\begin{equation}\label{eqn:stl}
    \psi := \top \ | \ \mathsf{p} \ | \ \neg \psi \ | \ \psi_1 \land \psi_2 \ | 
    \ \G_{[a, b]}\psi \ | \ \F_{[a, b]}\psi,
\end{equation}
where $\mathsf{p}$ is a predicate, and $\psi_1, \psi_2$ are STL formulae. Boolean operators for negation and conjunction are denoted by $\neg$ and $\land$, respectively. Symbols $\F$ and $\G$ represent the temporal operators - eventually and always, respectively. The time interval $[a, b]$ is defined with $a, b \in \R_0^+$, and satisfies $a \leq b$.
The satisfaction relation $(x,t)\models\psi$ denotes if a signal $x: \R_0^+ \rightarrow \R^n$, possibly a solution of \eqref{eqn:sysdyn_lin}, satisfies an STL formula $\psi$ at time $t$. The STL semantics \cite{STL_OdedMaler} for a signal $x$ is recursively given by:
\begin{align*}
    &(x,t)\models\mathsf{p} &&\Leftrightarrow \  h(x(t))\ge0, \\
    &(x,t)\models\neg\psi &&\Leftrightarrow \  \neg((x,t)\models\psi), \\
    & (x,t)\models\psi_1\land\psi_2 &&\Leftrightarrow \  (x,t)\models\psi_1\land (x,t) \models\psi_2, \\
    &(x,t)\models{\F}_{[a, b]}\psi&&\Leftrightarrow \  \exists t' \in [t+a,t+b], (x,t')\models\psi, \\
    &(x,t)\models{\G}_{[a, b]}\psi&&\Leftrightarrow \  \forall t' \in [t+a,t+b], (x,t')\models\psi.
\end{align*}

The formula $\psi$ satisfied by $x$, denoted by $x\models \psi$, if and only if $(x,0)\models\psi$. A key feature of STL is its robustness metric $\rho^{\psi}(x,t)$ \cite{STL_robust}, which quantifies how strongly a signal $x$ satisfies or violates a specification $\psi$. Specifically, the sign of the robustness metric gives the indication of satisfaction: $(x,t) \models \psi$ if $\rho^{\psi}(x,t) > 0$, and the magnitude indicates how strongly it satisfies. The robustness semantics for the mentioned STL formulae are recursively defined as
\begin{subequations}\label{eqn:stl_rho}
\begin{align}
    &\rho^{\mathsf{p}}(x,t)= h (x(t)), \label{eqn:stta}\\
    &\rho^{\neg\psi} (x,t)= -\rho^{\psi} (x(t)), \\
    &\rho^{\psi_1\land\psi_2} (x,t)= \min (\rho^{\psi_1} (x,t), \rho^{\psi_2} (x,t)),\\
    &\rho^{{\F}_{[a, b]}\psi} (x,t)= \max_{t_1 \in [t+a,t+b]} \rho^{\psi} (x,t_1),\\
    &\rho^{{\G}_{[a, b]}\psi} (x,t)= \min_{t_1 \in [t+a,t+b]} \rho^{\psi} (x,t_1).\label{eqn:stte}
\end{align}
\end{subequations}

We adopt the shorthand notation $\rho^{\psi}(x)$ to represent the robustness at $t = 0$ and consider a sufficiently large time horizon $[0,t_f]$ to verify the satisfaction of the formula $\psi$.

\cite{STL_Lip_Mangharam} and \cite{STL_Lip_Sadegh} have studied the continuity and smoothness of the STL robustness metric $\rho^\psi(x)$. Building on these works, we now formally establish its Lipschitz continuity.
\begin{thm}
\label{thm:rho_lip}
    Let $\psi$ be an STL formula over the time horizon $[0,t_f]$ with all predicate functions $h: \R^n \rightarrow \R$ being Lipschitz continuous. Then the robustness function $\rho^{\psi}(x)$ is Lipschitz continuous with respect to the sup norm:
    \begin{equation}\label{eqn:rho_lip}
        | \rho^{\psi}(x) - \rho^{\psi}(y) | \leq \mathcal{L}^{\psi}\|x-y\|_{\infty}, \ \mathcal{L}^{\psi} \in \R_0^+
    \end{equation}
    for all bounded continuous signals. 
\end{thm}
\begin{pf}
For atomic predicates, $\psi = \Big( h(x(0)) \geq 0 \Big)$, the robustness is $\rho^{\psi}(x) = h(x(0))$. If the predicate function, $h: \R^n \rightarrow \R$, is Lipschitz with constant $L_f$, then for all $x, y, : \R_0^+ \rightarrow \R^n$, we have
$$|\rho^{\psi}(x) - \rho^{\psi}(y)| = |h(x(0)) - h(y(0))| \leq L_f \|x - y\|_{\infty}.$$
Now, we proceed by structural induction on the STL formula $\psi$. For, negation, $\psi = \neg \psi$, we have $\rho^{\neg \psi}(x) = -\rho^{\psi}(x)$. This inherits the Lipschitz constant from $\psi$ as
$$|\rho^{\neg \psi}(x) - \rho^{\neg \psi}(y)| = |\rho^{\psi}(x) - \rho^{\psi}(y)|.$$
For the boolean operator conjunction $\psi = \psi_1 \land \psi_2$, robustness is defined as $\rho^{\psi}(x) = \min(\rho^{\psi_1}(x), \rho^{\psi_2}(x))$, with the $\min$ function being 1-Lipschitz continuous as 
$$|\min(a_1,a_2)-\min(b_1,b_2)| \leq \max(|a_1-b_1|,|a_2-b_2|).$$
So the overall Lipschitz constant becomes $\max(\mathcal{L}_{\psi_1}, \mathcal{L}_{\psi_2})$, where $\mathcal{L}_{\psi_1}, $ and $\mathcal{L}_{\psi_2}$ are the Lipschitz constants for $\rho^{\psi_1}(x)$ and $\rho^{\psi_2}(x)$, respectively. A similar analysis holds for disjunction via the $\max$ operator.
\\
For temporal operators, since $\min$ and $\max$ preserve Lipschitz continuity over bounded domains, and the inner functions are Lipschitz by the inductive hypothesis, the resulting robustness function $\rho^{\psi}(x)$ is also Lipschitz.
\\
Thus, by structural induction, $\rho^{\psi}(x)$ is Lipschitz continuous with constant $\mathcal{L}_\psi \in \R_0^+$, which depends on the formula structure and the Lipschitz constants of the predicates.
\end{pf}

\section{Resilience for STL Specifications}
Let the set of all trajectories of the system $\Sigma_l$ in Equation \eqref{eqn:sysdyn_lin}, starting from some $x_0 \in X$ at time $t = 0$, 
under arbitrary bounded disturbances $d: \R_0^+ \rightarrow \mathcal{B}_\infty(\varepsilon)$, be
\begin{equation}
    \Xi_\varepsilon(x_0) := \{ \xi_d(x_0) \mid \forall d: \R_0^+ \rightarrow \mathcal{B}_\infty(\varepsilon) \}.
\end{equation}

\begin{definition}[\cite{Resilience00}]
    Consider the dynamical system $\Sigma_l$ in Equation \eqref{eqn:sysdyn_lin}, an STL specification $\psi$ as in Equation \eqref{eqn:stl}, a set $X \subset \R^n$ and a point $x_0 \in X$. We define \textit{the resilience} of $\Sigma_l$ with respect to the initial condition $x(0) = x_0$ and the specification $\psi$ as a function $g_\psi: X \rightarrow \R_0^+ \cup \{+\infty\}$, with
    \begin{equation}
    \label{eq:quan}\hspace{-0.2em} g_\psi(x_0) = 
    \begin{cases}
    \sup\left\{\varepsilon\ge 0\,|\,\Xi_\varepsilon(x_0)\vDash \psi\right\}, &\Xi_0(x_0)\vDash\psi\\
    0, & \Xi_0(x_0)\nvDash\psi,
    \end{cases}
    \end{equation}
    where $\Xi_\varepsilon(x_0) \vDash \psi$ indicates
    $\xi_d(x_0) \vDash \psi, \forall d: \R_0^+ \rightarrow \mathcal{B}_\infty(\varepsilon).$
\end{definition}
For $X_0 \subset X$, $g_{\psi}(X_0) = \inf_{x_0 \in X_0}g_{\psi}(x_0)$.


Next, we show how this resilience metric can be approximated by computing bounds on system trajectories.

\section{Computation of Resilience}
This section focuses on deriving computational methods for resilience in linear dynamical systems. We begin by formulating bounds on the state trajectory of the perturbed linear system $\Sigma_l$ in Equation \eqref{eqn:sysdyn_lin}. Then, we demonstrate how these bounds enable the computation of resilience for STL specifications.

\subsection{Bounding State Trajectories}

\RD{Several approaches have been proposed in the literature to bound perturbed trajectories of dynamical systems. Classical reachability methods compute set-valued over-approximations of reachable states by propagating geometric sets, such as zonotopes or polytopes, through the system dynamics \cite{girard2005reachability}. In contrast, we aim to derive time-varying trajectory envelopes that are better suited for signal-level analysis and temporal logic evaluation. Since STL specifications are defined over entire signals, we focus on trajectory-level bounds rather than instantaneous reachable sets.} 
Other works, such as \cite{bound_m, bound_gamma}, study the existence of trajectory bounds, while \cite{Absolute_kofman} provide explicit bounds using the componentwise absolute value of the system matrices.

In this work, we discuss two methods for bounding the perturbed trajectory. The first approach leverages the Jordan canonical form decomposition of the system matrix, while the second relies on the element-wise absolute value of the system matrix to perform a worst-case analysis. By comparing these methods, we demonstrate that the Jordan-based approach yields tighter bounds than the absolute-value-based method.

In the absence of disturbances, the state trajectory follows the differential equation $\dot{x} = Ax$, with the corresponding trajectory at time $t$ given by 
$\xi_0(x_0,t) = e^{At}x_0.$
When disturbances are present, the state evolves according to Equation \eqref{eqn:sysdyn_lin}, and its solution at time $t$, $\xi_d(x_0,t)$ is 
\begin{equation}\label{eqn:alek_linear}
    \xi_d(x_0,t) = \xi_0(x_0,t) + \int_{0}^{t} e^{A(t-s)} d(s) \ds.
\end{equation}

Let $A$ be reduced to its Jordan canonical form $J$ using some invertible matrix transformation $P$ as $A = PJP^{-1}$. Then, the matrix exponential in \eqref{eqn:alek_linear} can be expressed as
$$e^{A(t-s)} = e^{PJP^{-1}(t-s)} = Pe^{J(t-s)}P^{-1}.$$

We first have the following auxiliary result.
\begin{lemma}\label{lem:Jordan}
    All elements of the matrix $e^{J_R(t-s)}$ are non-negative, i.e., $e^{J_R(t-s)} \succeq 0$, where $J_R$ is obtained by using real part of eigenvalues $J$.
\end{lemma}
\begin{pf}
    $J_R \!=\! \begin{bmatrix}
        J_{R,1} &          &    \\
            &   \ddots &    \\
            &          & J_{R,q}
    \end{bmatrix}$\!,\! 
    $J_{R,i} \!=\! \begin{bmatrix}
        \Rl(\lambda_i) & 1 & \\
            & \ddots & 1 \\
            & & \Rl(\lambda_i)
    \end{bmatrix} = \Rl(\lambda_i) I + N$,
    with $I$ as identity matrix and $N$ as nilpotent matrix of appropriate dimensions.
    We have $e^{tJ_{R,i}} = e^{t(\Rl(\lambda_i)I+N)} = e^{t\Rl(\lambda_i)I}e^{tN} = e^{t\Rl(\lambda_i)}Ie^{tN} = e^{t\Rl(\lambda_i)}e^{tN}$, where $e^{t\Rl(\lambda_i)} > 0$ as $\Rl(\lambda_i) \in \R$. 
    Now, since $N$ is nilpotent, it satisfies $N^n = 0$, for some $n \in \mathbb{N}$. Therefore,
    $$e^{tN} = I + tN + \frac{t^2N^2}{2!}  + \frac{t^3N^3}{3!} + \ldots + \frac{t^{n-1}N^{n-1}}{(n-1)!}.$$
    The matrix $N$ is a nilpotent matrix with $1$ on its superdiagonal and zeros elsewhere,
    $$\implies N^k \succeq 0, \forall k \in \mathbb{N} \implies e^{tN} \succeq 0.$$
    Therefore, $e^{tJ_{R,i}} = e^{t\Rl(\lambda_i)}e^{tN} \succeq 0$. And,
    \begin{align*}
    e^{tJ_R} &= \ e^{\begin{bmatrix}
        tJ_{R,1} &     &         \\
             & \ddots &    \\
            &     &         tJ_{R,q}
    \end{bmatrix}} = \begin{bmatrix}
        e^{tJ_{R,1}} &     &        \\
            & \ddots &    \\
            &        & e^{tJ_{R,q}}
    \end{bmatrix} \succeq 0 .
    \end{align*}
\end{pf}

Using this property, we now derive bounds on the perturbed trajectory $\xi_d(x_0,t)$.

\begin{thm}\label{thm:bd_jacob} 
    Consider the dynamical system $\Sigma_l$ in \eqref{eqn:sysdyn_lin} with a non-singular system matrix $A$. Under bounded disturbance, $d: \R_0^+ \rightarrow \mathcal{B}_\infty(\varepsilon)$, the perturbed trajectory $\xi_d(x_0,t)$ satisfies
    \begin{equation}\label{eqn:bd_jacob} 
    \xi_d(x_0,t) \in \xi_0(x_0,t) \oplus A_J(t) \mathcal{B}_{\infty}(\varepsilon),
    \end{equation} 
    where
    $A_J(t) := |P| \left( J_R^{-1}(e^{J_Rt}-\Id_n) \right) |P^{-1}|.$
\end{thm}
\begin{pf}
Substituting the Jordan decomposition into Equation~\eqref{eqn:alek_linear} gives
$$\xi_d(x_0,t) = \xi_0(x_0,t) + \int_0^t Pe^{J(t-s)}P^{-1} d(s)\,ds.$$
Taking elementwise absolute values and applying Lemma~\ref{lem:Jordan}, 
$$\bigl|\xi_d(x_0,t)-\xi_0(x_0,t)\bigr| \preceq \int_0^t |P|\,e^{J_R(t-s)}\,|P^{-1}|\,|d(s)|\,ds.$$
Using the disturbance bound $|d(s)|\preceq\varepsilon$ gives
$$\bigl|\xi_d(x_0,t)-\xi_0(x_0,t)\bigr| \preceq \int_0^t |P|\,e^{J_R(t-s)}\,|P^{-1}|\,\varepsilon\,ds.$$
Since $A$ is Hurwitz, $J_R$ is invertible, and the closed-form integral solution yields
$\int_0^t e^{J_R(t-s)}ds = J_R^{-1}(e^{J_Rt}-\Id_n).$
Substituting this into the bound above gives
$$\bigl|\xi_d(x_0,t)-\xi_0(x_0,t)\bigr|
\preceq |P|\,(J_R^{-1}(e^{J_Rt}-\Id_n))\,|P^{-1}|\,\varepsilon$$
Therefore,
\begin{align*}
    &\xi_0(x_0,t)-A_J(t)\varepsilon \preceq \xi_d(x_0,t) \preceq \xi_0(x_0,t)+A_J(t)\varepsilon
    \implies \xi_d(x_0,t)\in \xi_0(x_0,t)\oplus A_J(t)\,\mathcal{B}_\infty(\varepsilon). 
\end{align*}\qed
\end{pf}

For disturbance bounded by $d(t) \in \mathcal{B}_\infty(\varepsilon)$ with $\varepsilon > 0$, let this trajectory bound be equivalently expressed as a signal through $\overline{{\xi}}_\varepsilon (x_0) : \R_0^+ \rightarrow \mathbb{R}^n$, i.e., $\overline{{\xi}}_\varepsilon (x_0) := \xi_0(x_0) + A_J\varepsilon,$ 
\begin{align}\label{eqn:bd}
    \overline{\xi}_{-\varepsilon}(x_0,t) \le \xi_d(x_0,t) \le \overline{\xi}_{\varepsilon}(x_0,t), \quad \forall\, t \in \R_0^+.
\end{align}
\begin{lemma}\label{lem:rho_lip}
     {Consider an STL specification $\psi$ over the time interval $[0,t_f]$ as defined in \eqref{eqn:stl}, where the predicates are Lipschitz continuous.
     Let $d(t) \in \mathcal{B}_\infty(\omega)$ and the corresponding perturbed signal bound be given by $\overline{{\xi}}_\omega (x_0) := \xi_0(x_0) + A_J\omega.$ Then, $\rho^\psi(\overline{{\xi}}_\omega (x_0))$ is Lipschitz continuous w.r.t. $\omega$, with Lipschitz constant 
     $$\mathcal{L}_\omega := \mathcal{L}_\psi \max_{t \in [0,t_f]} \|A_J(t)\|,$$
     where $\mathcal{L}_\psi$ is the Lipschitz constant of the robustness measure from Theorem~\ref{thm:rho_lip}. This implies that, 
     $$|\rho^\psi(\overline{{\xi}}_{\omega_1} (x_0)) - \rho^\psi(\overline{{\xi}}_{\omega_2} (x_0))| \leq \mathcal{L}_\omega \|\omega_1 - \omega_2\|, \forall \omega_1, \omega_2>0.$$}
\end{lemma}
\begin{pf}
    {For an STL specification $\psi$, the Lipschitz constant $\mathcal{L}_\psi$ for $\rho^\psi$ can be obtained from Theorem~\ref{thm:rho_lip}. Accordingly, for disturbance bounds $\omega_1, \omega_2 >0$, we can write,
    \begin{align*}
    |\rho^\psi&(\overline{{\xi}}_{\omega_1}(x_0)) - \rho^\psi(\overline{{\xi}}_{\omega_2}(x_0))| 
        \leq \mathcal{L}_\psi \|\overline{{\xi}}_{\omega_1}(x_0) - \overline{{\xi}}_{\omega_2}(x_0)\|_{\infty} \\
        \leq& \mathcal{L}_\psi \| \xi_0(x_0) + A_J\omega_1 - \xi_0(x_0) - A_J\omega_2 \|_{\infty} 
        \leq \mathcal{L}_\psi \max_{t \in [0,t_f]} \|A_J(t)\| \| \omega_1 - \omega_2 \|
        \leq \mathcal{L}_\omega \|\omega_1 - \omega_2\|.
    \end{align*}\qed
    }
\end{pf}
We can also establish the bounds on the disturbed trajectory in Equation \eqref{eqn:alek_linear} using the element-wise absolute value of the system matrix 
\cite{Absolute_kofman}.

\begin{thm}\label{thm:bd_abs}
    Under bounded disturbance, $d: \R_0^+ \rightarrow \mathcal{B}_\infty(\varepsilon)$, the perturbed trajectory $\xi_d(x_0,t)$ satisfies
    \begin{align}\label{eqn:bd_abs}
    \xi_d(x_0,t) \in \xi_0(x_0,t) \oplus A_N(t)\mathcal{B}_{\infty}(\varepsilon),
    \end{align}
    where 
    $A_N(t) := \tilde{A}^{-1}(e^{\tilde{A}t}-1), \ \tilde{A} := |P||J||P^{-1}|.$
\end{thm}

{
\begin{pf}
Starting from Equation~\eqref{eqn:alek_linear}, take elementwise absolute values and applying the triangle inequality:
$$\bigl|\xi_d(x_0,t)-\xi_0(x_0,t)\bigr| \preceq \int_0^t \bigl|e^{A(t-s)}\bigr|\,|d(s)|\,ds.$$
Next, we bound the matrix exponential. Using the Taylor series expansion 
$e^{A\tau} = \sum_{k=0}^\infty A^k\tau^k/k!$
and the elementwise subadditivity and submultiplicativity of the absolute value ($|A^k|\preceq|A|^k$), we obtain 
$$\bigl|e^{A\tau}\bigr| =\Biggl|\sum_{k=0}^\infty \frac{A^k\tau^k}{k!}\Biggr| 
\preceq \sum_{k=0}^\infty \frac{|A|^k\tau^k}{k!} = e^{|A|\tau}.$$
Since $e^{|A|\tau}$ is a power series of the nonnegative matrix $|A|$, it has nonnegative entries. Substituting this bound gives
$$ \bigl|\xi_d(x_0,t)-\xi_0(x_0,t)\bigr| \preceq \int_0^t e^{|A|(t-s)}\,|d(s)|\,ds.$$
Using the disturbance bound $|d(s)|\preceq\varepsilon$ gives
\begin{align*}
    &\bigl|\xi_d(x_0,t)-\xi_0(x_0,t)\bigr|\preceq \int_0^t e^{|A|(t-s)}\,\varepsilon\,ds =:A_N(t)\,\varepsilon, \\
    \implies &\xi_0(x_0,t)-A_N(t)\varepsilon \preceq \xi_d(x_0,t) \preceq \xi_0(x_0,t)+A_N(t)\varepsilon, 
    \implies \xi_d(x_0,t)\in \xi_0(x_0,t)\oplus A_N(t)\,\mathcal{B}_\infty(\varepsilon).    
\end{align*}
\end{pf}
}



The bound in Theorem \ref{thm:bd_abs}, derived using the element-wise absolute value of the system matrix, is more conservative than the one in Theorem \ref{thm:bd_jacob}. This follows from the fact that $|P|J_R|P^{-1}| \preceq |P||J||P^{-1}|$, implying that the bounds in \eqref{eqn:bd_jacob} are always at least as tight as those in \eqref{eqn:bd_abs}, and often much tighter. The extent of this gap is illustrated in Figure~\ref{fig:comp_bound}. 
These bounds now form the basis for quantifying the system’s resilience to disturbances.

\begin{figure}
    \centering
    \includegraphics[width=0.8\textwidth]{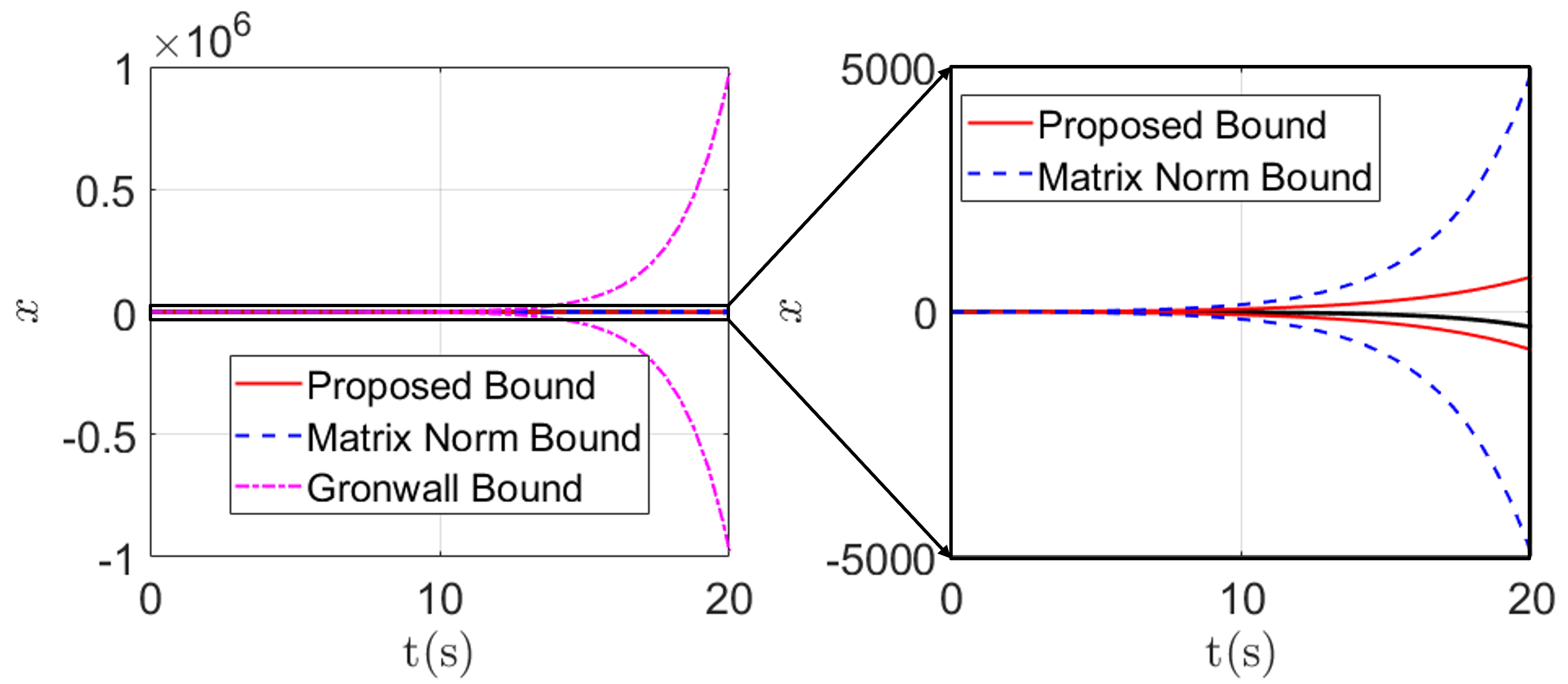}
    \caption{\RD{Comparison of bounds obtained using the proposed approach, element-wise absolute-value bounds, and Grönwall-based estimates.}}
    \label{fig:comp_bound}
\end{figure}

\subsection{Resilience for STL specifications}
We now use previously derived bounds to systematically compute a lower bound on resilience, establishing the maximum disturbance under which the system is guaranteed to maintain STL satisfaction.

\RD{
    From \eqref{eqn:bd}, the trajectories $\overline{\xi}_{-\varepsilon}(x_0)$ and $\overline{\xi}_{\varepsilon}(x_0)$ bound the system response $\xi_d(x_0)$ under any disturbance 
    $d:\mathbb{R}_0^+ \!\to\! \mathcal{B}_\infty(\varepsilon)$. 
    Since an intermediate trajectory within this bounded region may still violate the specification $\psi$, we evaluate the robustness $\rho^{\psi}$ over the entire family of trajectories
    $\{\overline{\xi}_{\omega}(x_0)\}_{\omega \in \mathcal{B}_\infty(\varepsilon)}$. 
    Ensuring nonnegative robustness across this family guarantees satisfaction of $\psi$ for all admissible trajectories $\xi_d(x_0)$, such that
    $
    \xi_d(x_0,t) \in \cup_{\omega \in \mathcal{B}_\infty(\varepsilon)}      \overline{\xi}_{\omega}(x_0,t),$ for all $t \in [0,t_f].
    $
    Therefore,
\begin{align}\label{eqn:impli}
        \min_{\omega \in [-\varepsilon,\,\varepsilon]} 
        \rho^{\psi}\big(\overline{\xi}_{\omega}(x_0)\big) > 0
        \;\implies\;
        \xi_d(x_0) \vDash \psi.
    \end{align}
}
The next theorem presents the main result of this paper.
\begin{thm}\label{thm:res_lin}
    Consider the system in \eqref{eqn:sysdyn_lin}, and an STL specification $\psi$ over the time interval $[0,t_f]$ as in \eqref{eqn:stl}. With the trajectory bounds $\overline{{\xi}}_{\omega} (x_0)$ derived in \eqref{eqn:bd_jacob} and the corresponding robustness of this signal bound, $\rho^\psi (\overline{{\xi}}_{\omega} (x_0))$, a lower bound on the resilience of system \eqref{eqn:sysdyn_lin} for the initial state $x_0\in X$ and specification $\psi$ is derived as
    \begin{gather}\label{eqn:res_lin}
        g_{\psi} (x_0) \geq \max \{ \varepsilon \geq 0 | \min_{\omega \in \mathcal{B}_{\infty}(\varepsilon)} \rho^\psi( \overline{{\xi}}_{\omega} (x_0) ) > 0\}.
    \end{gather}
\end{thm}


\begin{pf}
The robustness of the signal $\xi_d$ with respect to the STL specification $\psi$ is given by $\rho^\psi(\xi_d)$. To ensure satisfaction of $\psi$, the robustness $\rho^\psi(\xi_d)$ must be non-negative, i.e., $\rho^\psi(\xi_d) \geq 0$, for all $d: \R_0^+ \rightarrow \mathcal{B}_\infty(\varepsilon)$.

Using the trajectory bounds from \eqref{eqn:bd_jacob}, the worst-case robustness of the disturbed trajectory $\xi_d(x_0)$ over the disturbance set $\mathcal{B}_\infty(\varepsilon)$ is
$\min_{\omega \in \mathcal{B}_{\infty}(\varepsilon)} \rho^\psi\left( \overline{{\xi}}_{\omega} (x_0) \right).$
If this minimum robustness is non-negative, then the disturbed trajectory $\xi_d(x_0)$ satisfies $\psi$ for all $d: \R_0^+ \rightarrow \mathcal{B}_\infty(\varepsilon)$:
$$\xi_d(x_0) \vDash \psi, \forall d: \R_0^+ \rightarrow \mathcal{B}_\infty(\varepsilon).$$

This condition directly corresponds to the definition of the resilience metric \eqref{eq:quan}. Therefore, the resilience metric $g_{\psi}(x_0)$ for STL specification $\psi$ is given by
\begin{align*}
    g_{\psi} (x_0) = &\max\left\{\varepsilon \ge 0\,|\,\Xi_\varepsilon(x_0)\vDash \psi\right\} 
    =\max \{\varepsilon \ge 0\,|\, \xi_d(x_0) \vDash \psi, \forall d: \R_0^+ \rightarrow \mathcal{B}_\infty(\varepsilon) \}
\end{align*}
Consequently, from the implication \eqref{eqn:impli},
\begin{align*}
    g_{\psi} (x_0) \geq &\max \{ \varepsilon \geq 0 | \min_{\omega \in \mathcal{B}_{\infty}(\varepsilon)} \rho^\psi( \overline{{\xi}}_{\omega} (x_0) ) > 0 \}.
\end{align*}
Thus, the computed resilience metric represents a lower bound on the maximum disturbance level $\varepsilon$ that ensures STL satisfaction under worst-case disturbances. \qed
\end{pf}
{
\begin{remark}
    The resilience expression in \eqref{eqn:res_lin} computes the minimum robustness over all disturbance signals $\omega \in \mathcal{B}_\infty(\varepsilon)$. This allows us to determine the largest disturbance envelope $\mathcal{B}_\infty(\varepsilon)$ such that the entire set of trajectories within it satisfies the specification. Importantly, due to the time-varying nature of the envelope bounds (as shown in Figure~\ref{fig:comp_bound}), the robustness is evaluated over entire trajectories, thereby preserving the temporal semantics of STL.
\end{remark}
}

Since the computation requires exploring over the continuous disturbance space $\mathcal{B}_{\infty}(\varepsilon)$, the problem inherently involves an infinite number of constraints. To address this, we sample $M$ data points, $\left\{\omega_r | r \in \{1,\ldots,M\} \right\},$ from the disturbance space $\mathcal{B}_{\infty}(\varepsilon)$. Around each sample $\omega_r$, we define $\delta$-balls $\Omega_r$, such that for all points in the disturbance space $\omega \in \mathcal{B}_{\infty}(\varepsilon)$, there exists a point $\omega_r$ which satisfies $\|\omega-\omega_r\| \leq \delta$. This ensures that the union of all these balls forms a superset of the disturbance space, i.e., $\bigcup_{r=1}^{M}\Omega_r \supset \mathcal{B}_{\infty}(\varepsilon)$.

Now, consider the scenario problem:
\begin{align}\label{eqn:SOP}
    &\max_{\varepsilon} \min_{\eta} \quad \eta \\
    &\textrm{s.t. } \forall \omega_r \in \mathcal{B}_{\infty}(\varepsilon), \ -\rho^\psi( \overline{\xi}_{\omega_r} (x_0)) \leq \eta \notag. 
\end{align}

\begin{thm}\label{thm:SOP}
    Let $(\varepsilon^*, \eta^*)$ be the solution to the scenario problem. Then, for all $\omega \in \mathcal{B}_{\infty}(\varepsilon^*)$, the condition $\rho^\psi( \overline{\xi}_{\omega} (x_0)) > 0 $ is guaranteed if
    \begin{equation}
        \eta^* + \mathcal{L}_\omega\delta \leq 0, 
    \end{equation}
    where $\mathcal{L}_\omega$ denotes the Lipschitz constant of $\rho^\psi( \overline{\xi}_{\omega} (x_0))$, as established in Lemma \ref{lem:rho_lip}.
\end{thm}
\begin{pf}
At the sampled data points $\omega_r \in \mathcal{B}_{\infty}(\varepsilon),$ for all $r \in \{1, \ldots, M\}$, we have $-\rho^\psi( \overline{\xi}_{\omega_r} (x_0)) \leq \eta^*$. Now, for the continuous patch, i.e., for all $\omega \in \mathcal{B}_{\infty}(\varepsilon)$,
\begin{align*}
    -\rho^\psi( \overline{\xi}_{\omega} (x_0)) 
    &\!=\! -\rho^\psi( \overline{\xi}_{\omega} (x_0)) + \rho^\psi( \overline{\xi}_{\omega_r} (x_0)) - \rho^\psi( \overline{\xi}_{\omega_r} (x_0)) \\
    &\!\leq\! \| \rho^\psi( \overline{\xi}_{\omega} (x_0)) - \rho^\psi( \overline{\xi}_{\omega_r} (x_0)) \| - \rho^\psi( \overline{\xi}_{\omega_r} (x_0)) \\
    &\!\leq\! \mathcal{L}_\omega\| \omega - \omega_r \| + \eta^* \leq \mathcal{L}_\omega\delta + \eta^* < 0.
\end{align*}
This implies $\rho^\psi( \overline{\xi}_{\omega} (x_0)) > 0$ for all $\omega \in \mathcal{B}_{\infty}(\varepsilon)$. \qed
\end{pf}

\begin{remark}  
    The tightness of the resilience metric obtained from Theorem \ref{thm:SOP} can be further examined by establishing a probabilistic relationship between the optimal values of the scenario problem in Equation \eqref{eqn:SOP} and the robust problem in Equation \eqref{eqn:res_lin}. This connection follows from the theoretical guarantees provided by \cite{esfahani2014performance}, which establish a bridge between the optimal values of a robust optimization problem and its scenario-based approximation, and have been utilized in subsequent works by \cite{kazemi2024data} and \cite{nejati2023formal}.
\end{remark}  


\section{Application to Nonlinear systems}\label{sec:nonlin}
Next, we extend the resilience metric computation to a class of nonlinear systems. 
\RD{While the main framework of this paper is developed for linear systems, here we provide a local extension to nonlinear dynamics via linearization around an equilibrium point and treating the higher-order terms as a bounded perturbation. The analysis is local in nature and assumes that trajectories remain within a compact set $X$ around the equilibrium point, in which the linearization is valid and the Hessian is bounded.}
\begin{thm}
    Consider the nonlinear system $\Sigma_{nl}$ in~\eqref{eqn:sysdyn_nl}, with an STL specification $\psi$ over the time interval $[0,t_f]$, as in~\eqref{eqn:stl}. 
    Suppose the system is linearized around an equilibrium point $x_e$, where $f(x_e) = 0$, with the Jacobian matrix defined as $A := \frac{\partial f}{\partial x} |_{x=x_e}$. 
    \RD{Assume that $f$ is twice continuously differentiable and that the trajectories remain inside a compact set $X$ in which the Hessian is bounded.}
    Let the trajectory bounds of the linearized system be given by $\overline{\xi}_\omega$, derived in \eqref{eqn:bd_jacob}, and the corresponding robustness of this signal bound be denoted by $\rho^\psi (\overline{\xi}_\omega)$. Then, the lower bound on the resilience metric is
    \begin{align}
        &g_{\psi} (x_0) \geq \max \{ \varepsilon \geq 0 | \min_{\omega \in \mathcal{B}_{\infty}(\varepsilon)} \rho^\psi( \overline{{\xi}}_{\omega} (x_0) ) > 0\} - \overline{\delta}, \\
        &\text{with, } \overline{\delta} \geq \frac{1}{2}L_H\|x-x_e\|^2, \forall x \in X \notag,
    \end{align}
    provided that the resilience of the linearized system is greater than $\overline{\delta}$, i.e.,
    $$\max \{ \varepsilon \geq 0 | \min_{\omega \in \mathcal{B}_{\infty}(\varepsilon)} \rho^\psi( \overline{{\xi}}_{\omega} (x_0) ) > 0\} \geq \overline{\delta}.$$
    Here $L_H$ is an upper bound on the norm of the Hessian $H(x)$, i.e.,  $\|H(x)\| \leq L_H$ for all $x \in X$.
\end{thm}
\begin{pf}
Linearizing the nonlinear system around $x_e$, we approximate the system dynamics as
$$\dot{x} = f(x) + d(t) = A(x-x_e) + d(t) + \delta(t),$$
where $A$ is the Jacobian of $f(x)$ at $x_e$.
Using Taylor's theorem, the remainder term satisfies
\begin{align*}
    \delta(t) \!=\! \frac{1}{2}(x-x_e)^\top H(x) (x-x_e)
              \!\leq\! \frac{1}{2}L_H\|x-x_e\|^2 \!\leq\! \overline{\delta}.
\end{align*}
By using the resilience metric for the linearized system in Theorem \ref{thm:res_lin} and accounting for the approximation error introduced by linearization, the lower bound of the resilience metric for the nonlinear system is
$$g_{\psi} (x_0) \geq \max \{ \varepsilon \geq 0 | \min_{\omega \in \mathcal{B}_{\infty}(\varepsilon)} \rho^\psi( \overline{{\xi}}_{\omega} (x_0) ) > 0\} - \overline{\delta}.$$\qed
\end{pf}
We demonstrate this approach through a numerical example in Section \eqref{Nonlin_case_study}.

\begin{remark}
\RD{
The nonlinear extension should be interpreted as a local approximation around the equilibrium point rather than a global guarantee. The above result provides a conservative local bound due to the linearization error term $\overline{\delta}$, which depends on $L_H$ and can be obtained either analytically from the system dynamics or numerically over the region of interest $X$.
}

\RD{
One possible extension is to use piecewise linear approximations over a state space region, where local trajectory bounds are computed and combined to obtain a less conservative global estimate. For more gener cases, set-based reachability techniques, such as zonotope methods \cite{girard2005reachability, zonoNonlin}, could be used to compute tighter bounds for nonlinear dynamics. Integrating such approaches with the proposed resilience framework, especially for more general nonlinear systems is an important direction for future work.
}
\end{remark}

\section{Case Studies}
We present four case studies to demonstrate the effectiveness of the proposed approach in quantifying the maximum disturbance a system can tolerate while still satisfying the given STL specification.

\begin{figure}
    \centering
     \begin{subfigure}[b]{0.46\textwidth}
         \centering
         \includegraphics[width=\textwidth]{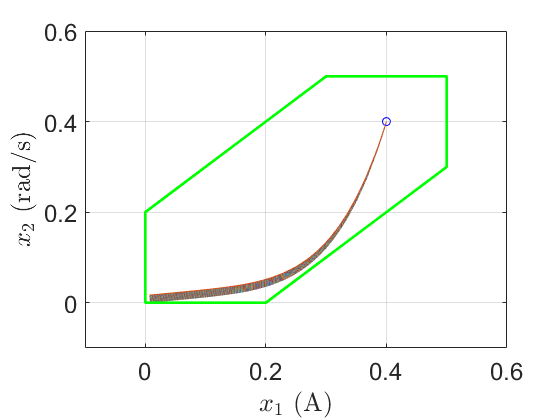}
         \caption{$\psi_1 = \square_{[0,20]} (G_Sx \leq H_S)$}
     \end{subfigure}
     \hfill
     \begin{subfigure}[b]{0.46\textwidth}
         \centering
         \includegraphics[width=\textwidth]{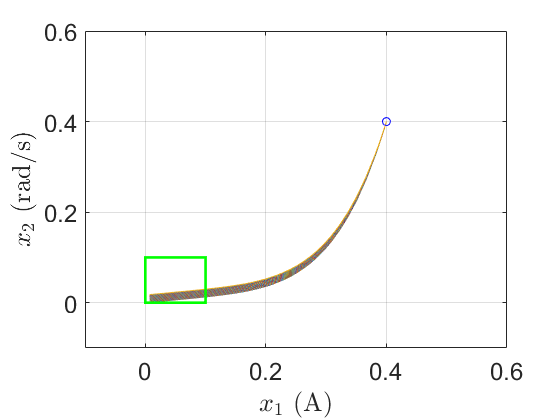}
         \caption{$\psi_2 = \square_{[0,20]} (G_Rx \leq H_R)$}
     \end{subfigure}
    \caption{Evolution of 1000 perturbed trajectories in the DC Motor case. Green lines mark the safe and target set boundaries, showing all trajectories satisfy $\psi_1$ and $\psi_2$.}
    \label{fig:lincon}
\end{figure}

\subsection{DC Motor}
In this example, we consider a DC motor system \cite{DCMotor} with dynamics:
\begin{align*}
    \dot{x}_1 \!=\! - R/L x_1 \!-\! K/L x_2 \!+\! d_1, \
    \dot{x}_2 \!=\! K/J x_1 \!-\! b/J x_2 \!+\! d_2,
\end{align*}
where $x_1$ and $x_2$ are the armature current and shaft angular velocity, respectively, and $d = [d_1, d_2]^\top$ is the additive disturbance. The system parameters are $R = 1 \Omega$, $L = 0.5 H$, $J = 0.01 kg.m^2$, $b = 0.1 N.m.s$, and $K = 0.01$. Starting from $x_0 = [0.4, 0.4]^\top$, we evaluate the resilience metric for the system under two STL specifications: 

\noindent\textbf{Safety Specification:} The system must remain within a safe set for all $t \in [0,20]$s:
\begin{gather*}
    \psi_1 = \square_{[0,20]} (G_Sx \leq H_S), \\
    G_S = \begin{bmatrix}
        -1, 1, 0, 0, -1, 1\\
        0, 0, -1, 1, 1, -1
    \end{bmatrix}^\top, \
    H_S = \begin{bmatrix}
        0, 0.5, 0, 0.5, 0.2, 0.2
    \end{bmatrix}^\top.
\end{gather*}
The maximum disturbance for $\psi_1$ is $0.2089.$

\noindent\textbf{Reachability Specification:} The system must reach a target set at least once in $t \in [15,20]$:
\begin{gather*}
    \psi_2 = \lozenge_{[15,20]} (G_Rx \leq H_R), \\
    G_R = \begin{bmatrix}
        -1, 1, 0, 0\\
        0, 0, -1, 1
    \end{bmatrix}^\top, \
    H_R = \begin{bmatrix}
        0, 0.1, 0, 0.1
    \end{bmatrix}^\top.
\end{gather*}
The maximum admissible disturbance for $\psi_2$ is $0.0187$.

The resilience metric can also be computed using the bounds established in Theorem \ref{thm:bd_abs}. However, the absolute-value-based method tends to overestimate system fragility, yielding significantly smaller values ($2.3178\times10^{-10}$ for $\psi_1$ and $6.5842\times10^{-5}$ for $\psi_2$) compared to the proposed approach, which provides a more accurate and practical measure of disturbance tolerance.


Figure \ref{fig:lincon} illustrates 1000 perturbed trajectories satisfying the safety and reachability specifications, $\psi_1$ and $\psi_2$. 



\begin{figure}[t]
     \centering
     \includegraphics[width=0.6\textwidth]{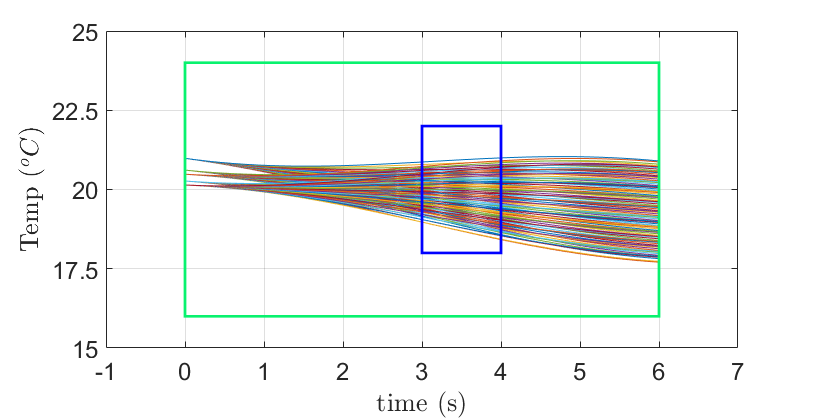}
     \caption{Evolution of 1000 perturbed trajectories in the temperature regulation case. Green and blue correspond to the specifications for $T_{S_1}$ and $T_{S_2}$, respectively.}
     \label{fig:temp}
\end{figure}

\subsection{Temperature Regulation}
Here, we model temperature dynamics \cite{NAHS} in a circular building with five interconnected rooms:
\begin{align*}
    \dot{T}_i = \alpha (T_{i+1} + T_{i-1} - 2T_{i}) + \beta (T_{e} - T_{i}), i = \{1, \ldots, 5\},
\end{align*}
where $\alpha = 0.1$ and $\beta = 0.01$ represent the heat exchange rate between adjacent rooms, and the influence of the external temperature $T_e$, respectively. $T_e$ acts as an external disturbance to the system.

The desired temperature regulation is defined by two constraints: each room's temperature must remain within $T_{S_1} = [16^oC, 24^oC]^5$ over the entire time period $[0,6]$ hours; and during the interval $[3,4]$ hours, each room’s temperature must be maintained within a narrower range of $T_{S_2} = [18^oC, 22^oC]^5$. This objective is captured by the STL specification
\begin{align*}
    \psi(T) = \square_{[0,6]} (T \in T_{S_1}) \wedge \square_{[3,4]} (T \in T_{S_2}).
\end{align*}

The results indicate that the system can tolerate a maximum of $\pm0.4017^oC$ external fluctuations while satisfying $\psi$. Figure \ref{fig:temp} illustrates the results of 1000 simulations where external temperature variations are randomly sampled from the interval $[-0.4017^oC, 0.4017^oC]$.

\begin{figure}[t]
     \centering
     \includegraphics[width=0.5\textwidth]{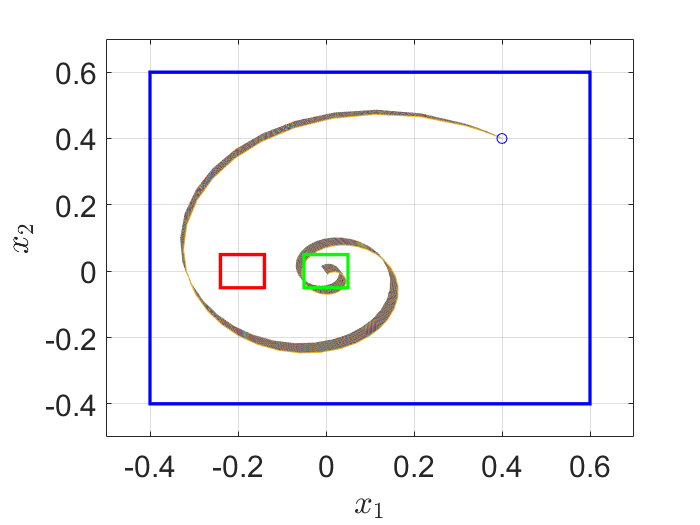}
    \caption{\RD{Evolution of 1000 perturbed trajectories starting from $x_0 = [0.4, 0.4]^\top$. Red, green, and blue denote the boundaries of the sets $U$, $T$, and $S$, respectively.}}
    \label{fig:num_example}
\end{figure}

\begin{figure*}
    \centering
    \begin{subfigure}[b]{0.6\textwidth}
         \centering
         \includegraphics[width=\textwidth]{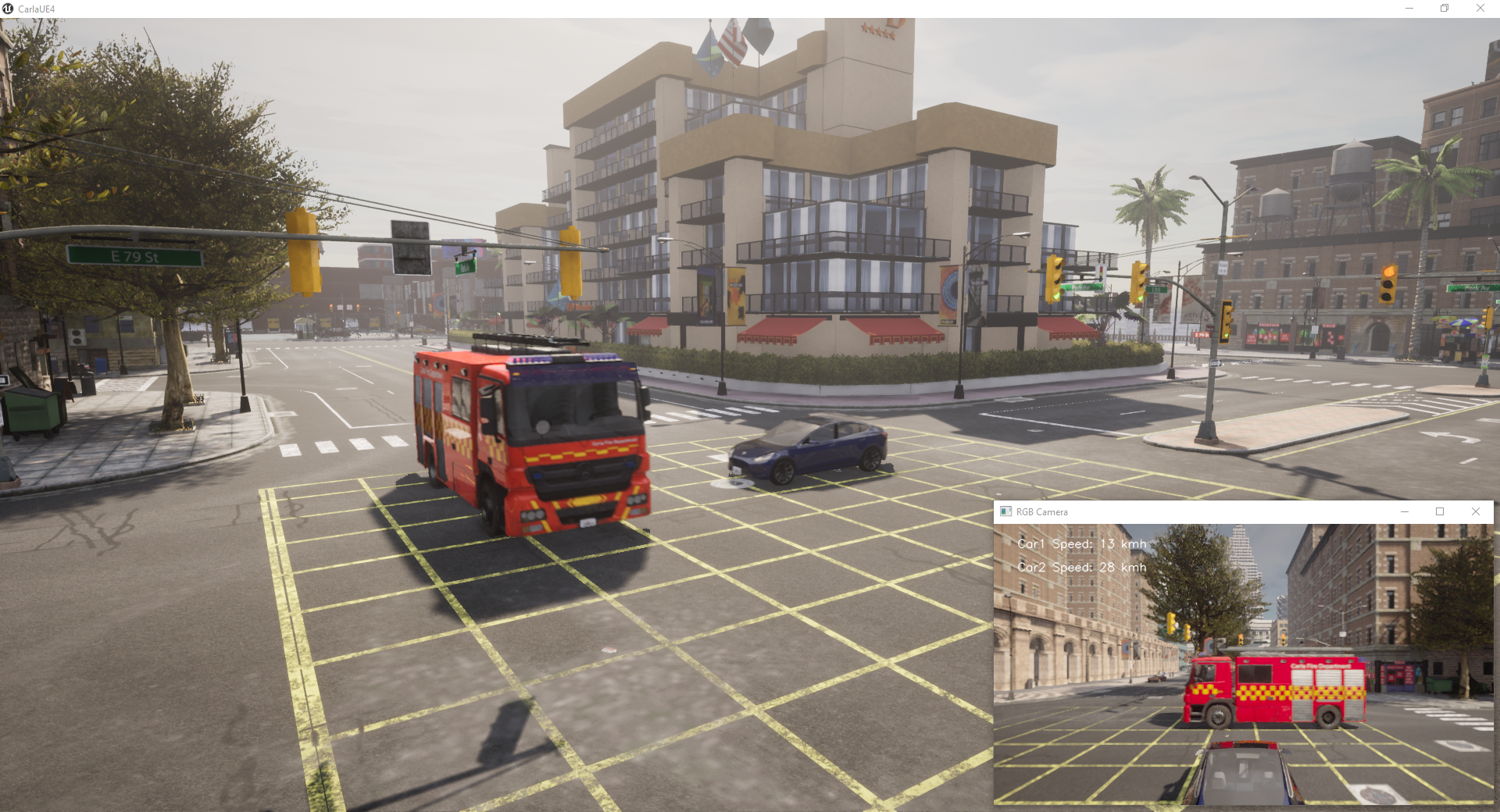}
         \caption{}
     \end{subfigure}\\
     \hfill
     \begin{subfigure}[b]{0.33\textwidth}
         \centering
         \includegraphics[width=\textwidth]{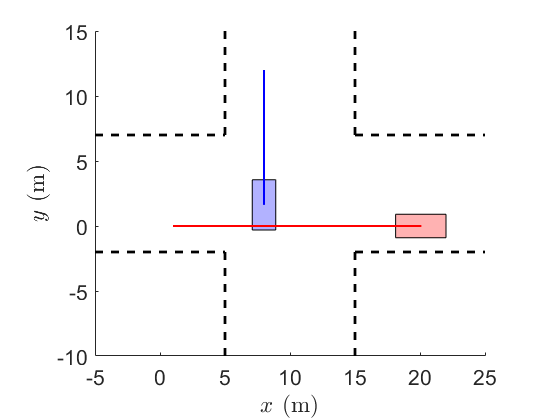}
         \caption{}
     \end{subfigure}
     \hfill
     \begin{subfigure}[b]{0.33\textwidth}
         \centering
         \includegraphics[width=\textwidth]{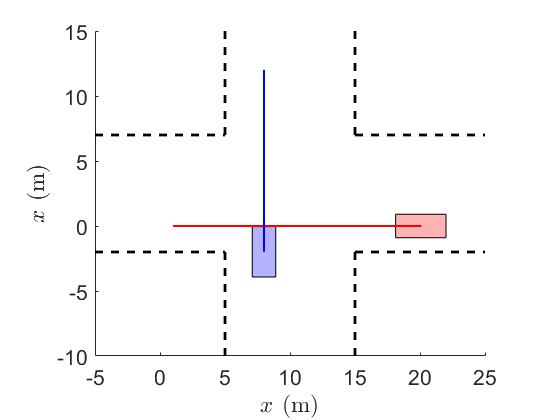}
         \caption{}
     \end{subfigure}
     \hfill
     \begin{subfigure}[b]{0.33\textwidth}
         \centering
         \includegraphics[width=\textwidth]{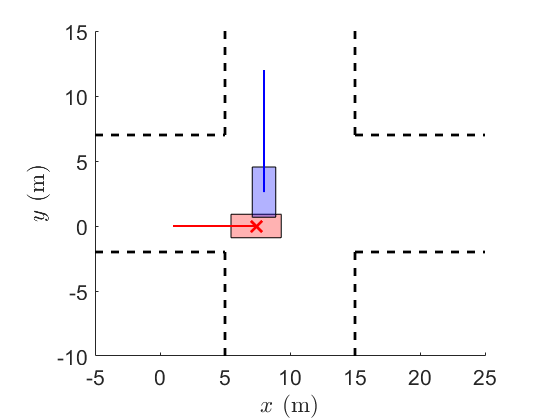}
         \caption{}
     \end{subfigure}
    \caption{(a) Carla simulation environment with the fire truck $V_1$ (red) and autonomous vehicle $V_2$ (blue) at an intersection. \href{https://youtu.be/-R9lWPYW560}{Video Link} (b) Nominal behavior with no disturbance. (c) Collision avoidance when the disturbance is bounded by the computed maximum disturbance. (d) Collision when the disturbance exceeds the computed maximum disturbance.}
    \label{fig:cross}
\end{figure*}

\subsection{Nonlinear Example}
\label{Nonlin_case_study}
This case study examines the nonlinear dynamical system
\begin{align*}
    \dot{x}_1 \!=\! -0.5x_1 \!-\! 2x_2 \!+\! 0.1x_1^2 \!+\! d_1, 
    \dot{x}_2 \!=\! 2x_1 \!-\! 0.5x_2 \!+\! 0.1x_2^2 \!+\! d_2.   
\end{align*}
The eigenvalues of the linearized system are $\lambda_{1,2} = -0.5 \pm 2i$.
Since the real part of the eigenvalues is negative, the system exhibits damped oscillations. However, when subjected to external disturbances, the system's behavior can shift between stable and unstable regimes, making the computation of resilience particularly significant.

\RD{
We evaluate the system under the STL specification:
$$\psi = \lozenge_{[0,7]}\square_{[0,5]} (x \in T) \wedge \square_{[0,12]} \neg(x \in U) \wedge \square_{[0,12]} (x \in S).$$
This condition requires the system to eventually reach a target region $T =[-0.05,0.05] \times [-0.05,0.05]$ and remain inside it for at least $5$ sec within the time window $[0,7]$, while always avoiding the unsafe set $U = [-0.24,-0.14] \times [-0.05,0.05]$ and always remaining inside the bounded region $S = [-0.4,0.6] \times [-0.4,0.6]$ for the interval $[0,12]$ sec.
The maximum admissible disturbance is $0.0396.$
Figure \ref{fig:num_example} presents the results of 1000 simulated trajectories, where disturbances are randomly sampled from the interval, computed using this bound. 
}


{\subsection{Intersection Collision Avoidance}}
\label{Vehicle_case}
This case study demonstrates the practical application of the resilience metric in a safety-critical traffic scenario. It involves a fire truck $V_1$ and an autonomous passenger vehicle $V_2$ at a $90^o$ intersection. The kinematic single-track model adopted from \cite{commonRoad} is used for vehicle dynamics and is linearized around the nominal trajectory. We consider an additive disturbance acting on the velocity of the autonomous vehicle $V_2$.


The collision avoidance specification $\psi$ requires the distance between the two vehicles, $\mathsf{dist}(V_1, V_2)$ to always remain at least a minimum safe distance $D_{\min} = 3m$ over the time horizon $t \in [0,10]s$:
$$\psi = \Box_{[0,10]}\big( \mathsf{dist}(V_1, V_2) \geq D_{\min} \big).$$

The computed maximum admissible disturbance for this task is $g_{\psi}(x_0) = 2.0871$, which is the maximum additive velocity disturbance the autonomous vehicle can tolerate while formally guaranteeing the collision-avoidance specification $\psi$.

Figure \ref{fig:cross} presents the results for three different disturbance scenarios. First, the nominal behavior (no disturbance), where the autonomous vehicle successfully avoids the fire truck. Second, the case where the additive velocity disturbance is bounded by the resilience metric, showing that the collision avoidance is still maintained, thus validating the formal safety guarantee. Finally, the case where the disturbance is $3$~m/s, which exceeds $2.0871$~m/s, leading to a collision between the vehicles. 
Videos of the corresponding Carla \cite{CARLA} simulations can be found at this link:
\href{https://youtu.be/-R9lWPYW560}{https://youtu.be/-R9lWPYW560}.

\section{Conclusion and Future Work}
\label{sec:conclusion}

\vspace{-0.2cm}

In this work, we introduced a framework for quantifying a lower bound on the resilience metric in continuous-time systems with respect to STL specifications, providing a formal approach to determine the maximum disturbance a system can tolerate while maintaining specification satisfaction. We derived trajectory bounds and then employed scenario optimization for computing the maximum admissible disturbance in continuous-time systems. The proposed framework was validated through case studies across various systems, demonstrating its practical applicability. 
\RD{
We also aim to explore how interval-based STL monitoring methods could be combined with the proposed trajectory bounds to evaluate robustness within deviation polytopes.
}

\bibliographystyle{unsrt}
\bibliography{sources}  

\end{document}